\documentclass{llncs}
\usepackage{url}
\usepackage{amsmath, amssymb}
\usepackage{algorithmic}
\usepackage{algorithm}

\usepackage{latexsym}
\usepackage{graphicx}
\usepackage{wrapfig}

\usepackage{cite}

\begin{document}
\title{Evader Interdiction and Collateral Damage}

\author{
Matthew~P.~Johnson$^1$\thanks{This work was performed in part while visiting Los Alamos National Laboratory.}~~
 Alexander Gutfraind$^2$\\~\\
$^1$ Pennsylvania State University.\\
$^2$ Theoretical Division, Los Alamos National Laboratory.}
\institute{}
\date{}
\maketitle

\begin{abstract}
In network interdiction problems, evaders (e.g., hostile agents or data packets) may be moving through a network 
towards targets and we wish to choose locations for sensors in order to intercept the evaders before they reach their destinations. 
The evaders might follow deterministic routes or Markov chains, or they may be {\em reactive}, i.e., able to change their routes in order to avoid sensors placed to detect them.
The challenge in such problems is to choose sensor locations economically, balancing security gains with costs, including the inconvenience sensors inflict upon innocent travelers. We study the objectives of 1) maximizing the number of evaders captured when limited by a budget on sensing cost and 2) capturing {all} evaders as cheaply as possible.

We give optimal sensor placement algorithms for several classes of special graphs and hardness and approximation results for general graphs, including for deterministic or Markov chain-based and reactive or oblivious evaders.
In a similar-sounding but fundamentally different problem setting posed by \cite{Rubinstein} where both evaders {\em and} innocent travelers are reactive, we again give optimal algorithms for special cases and hardness and approximation results on general graphs.

\keywords{network interdiction, bridge policy, submodular set cover, Markov chain, minimal cut}
\end{abstract}

\section{Introduction}
In network interdiction problems, one or more {\em evaders} (e.g., smugglers or terrorists, or hostile data packets) travel through a network, beginning at some initial locations and attempting to reach some targets. Our goal is to stop them. We do so by placing {\em sensors} on nodes 
in hopes that most or all the evaders will pass by a sensor and thus be captured (or {\em intercepted}) before reaching their destinations. 
We take as given the evader movement dynamics, which may be either deterministic (each evader specified by a path from source to target) or stochastic, e.g. each evader specified by a Markov chain 
whose states are the nodes of the network.  Evader $e_i$ induces a subgraph $G_i \subseteq G$ in which she roams, according to the probabilities specified by her Markov chain.
An unreactive or {\em oblivious} evader \cite{Gutfraind09unreactive} behaves the same regardless of the choice of sensor locations (or {\em interdiction} sites), and so her set of possible routes can be construed as objects we wish to pierce.

We try to make economical use of the sensors---i.e., to balance the benefits of security (the interdiction of many or all evaders) with the total cost (widely defined) of doing so. The 
cost of placing a sensor at a given node can incorporate the cost of the device itself, the effort or danger involved in performing the placement, and the inconvenience it causes to any innocent travelers subjected to it.
If traffic flow estimates on the graph's edges are known for both evaders and innocent travelers, then it is natural to try to place sensors where they will intercept many evaders but inconvenience few innocents. If a sensor acts as a checkpoint, capturing the evaders but examining and then letting pass the innocents, then the inconvenience cost can be incorporated directly into the node's sensor placement cost since placing two sensors on an innocent's path inconveniences her twice. In this model we study two natural objectives: 1) maximizing the (expected, weighted) number of evaders captured while respecting a budget on sensing cost, and 2) capturing {\em all} evaders (with probability 1) as cheaply as possible. In the latter case evaders may be {\em reactive}, i.e., able to observe the  sensor locations and choose a different path in $G_i$. Regardless, $e_i$ is guaranteed to be captured only if her target node is separated from all her source nodes {\em within subgraph $G_i$}.
We solve these problems optimally in several special graph settings and give hardness and approximation results in general settings.

\begin{wrapfigure}[18]{R}{0.375\columnwidth}
\vskip -.65cm
\centering \includegraphics[width=.25\textwidth]{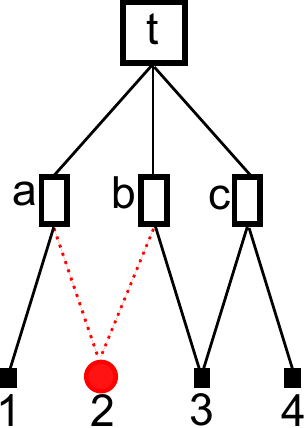} \vskip -.25cm
\caption{A bridges problem instance represented as network interdiction with three intermediate nodes corresponding to bridges. An innocent begins at node 2 and evaders begin at nodes 1,3,4. \label{fig:bridges}}
\end{wrapfigure}
In contrast, allowing the innocents {\em also} to react to sensor locations 
changes the character of the problem significantly. In this setting we study a special case of the problem which was posed by Glazer \& Rubinstein \cite{Rubinstein}, motivated by the following scenario: there are a collection of bridges crossing a river, with each traveler $p$ restricted to using some set $\sigma(p)$ of bridges (because of $p$'s preferences or geography, say), and the task is to decide
which bridges to open and close.
This can be viewed as a special case of our network setting in which every travel path is of length 2 but with the restriction that sensors cannot be placed on a traveler's start node (see Fig.~\ref{fig:bridges}). Note that in this special case, sensors can also be viewed as roadblocks, in the sense that placing a sensor on a node effectively means deleting the node from the network for evader and innocent alike.

This change yields a setting in which the problem instance is specified by a set system with real-valued elements that may be either positive or negative, corresponding to the value or cost (respectively) of capturing evaders or blocking innocents. Several possible objective functions could be considered, such as capturing all evaders while blocking as few innocents as possible or capturing as many evaders as possible given a budget allowing a certain number of blocked innocents.
Unfortunately, the former is precisely the Red-Blue Set Cover problem, which is ``strongly inapproximable'' (hard to approximate with factor $\Omega(2^{\log^{1-\epsilon}} m)$ for any $\epsilon>0$ (where $m$ is the number of sets, or bridges) unless $NP \subseteq DTIME(m^{\text{polylog}(m)})$) \cite{Peleg07}; the latter turns out to be harder still (see Appendix \ref{app:hardness}).
Instead, we next study objectives suggested by \cite{Rubinstein} that combine the two goals into a single score, where captured (i.e. unsuccessful) innocents and uncaptured (i.e. successful) evaders can be construed as false positives ($FP$) and false negatives ($FN$), respectively: maximize the ``net flow'' ($TN-FN$) or minimize the total errors ($FP+FN$).  Although the $TN-FN$ model turns out to be hard to approximate with factor $n^{1-\epsilon}$ (once again, see Appendix \ref{app:hardness}), we obtain nontrivial approximation results for the $FP+FN$ setting.

\noindent \textbf{Contributions.}
With oblivious evaders and innocents, we solve the budgeted problem optimally in path and cycle graphs. With oblivious innocents (and evaders either oblivious or not), we solve the full interdiction problem optimally in paths, cycles, and trees. (In the edge sensor model, full interdiction is 2-approximable, which is the optimal approximation factor assuming UGC.)

In general graphs, we give hardness results including showing that the budgeted problem is NP-hard with even one Markovian evader, strengthening the hardness result of \cite{Gutfraind09umek} which held only for two or more such evaders. In contrast, we show that full interdiction with a single evader is in P. With $m$ possible evader paths, the problem is $H_m$-approximable (where $H_m=\sum_{i=1}^m\frac{1}{i}$), which is essentially optimal, given certain complexity-theoretic assumptions.

When both evaders and innocents are reactive, we optimally solve a special case in which the graph and travelers' sets of paths
can be represented by {\em bridges}  and {\em convex bridge sets}, respectively. We also show that $FP+FN$ is approximable with factor one plus the maximum size of any innocent's bridge set.

\noindent \textbf{Related Work.} 
The problems analyzed here belong to a large class of discrete optimization problems, 
collectively termed Network Interdiction \cite{Mcmasters-1970-optimal,Corley-1982-most,Pan-2003-models,Gutfraind11markovian}. 
They are motivated by applications such as supply chains, electronic sensing, and counter-terrorism  and
relate to classical optimization problems like Set Cover and Max Coverage.
Our setting of budgeted interdiction with deterministic evaders on the path graph can be solved by a complicated algorithm given by \cite{Megiddo83}, but we present a much simpler algorithm. Recent work on Set Cover with submodular costs \cite{KoufogiannakisY09,IwataN09} applies to some of our settings.
Previous work on the Unreactive Markovian Evader (UME) interdiction problem (maximizing the expected number of Markov chain-based evaders captured with $B$ sensors) showed that it is already NP-hard with just two evaders \cite{Gutfraind09umek} and that it is $\frac{e}{e-1}$-approximable by the natural greedy algorithm \cite{Gutfraind09unreactive}, which is the optimal approximation factor (we prove this for completeness in Proposition \ref{prop:budgetedhardness}).

Other evader models have been studied such as the Most Vital Nodes Problem, in which
the task is to delete a set of nodes in order to maximize the weight of the shortest path 
from source to destination \cite{Corley-1982-most,Bar-noy-1995-complexity} or 
to decrease the maximum flow \cite{Corley1974,Ratliff1975}, 
both of which could be construed as frustrating an evader's progress. 
Such evaders are {\em reactive} in the sense that the routes they take 
are modified based on the set of available edges or nodes. 
In \cite{Gutfraind11markovian}, an intermediate model was studied in which the evader follows a parametrized generalization of shortest path and random walk.

The Bridges problem was introduced by Glazer \& Rubinstein \cite{Rubinstein} in an economics context, primarily motivated in terms of strategies for a listener to accept good arguments and reject bad arguments. In this setting, (positive/negative) states correspond to (positive/negative) people and allowing oneself to be persuaded by a statement corresponds to opening a bridge.

\section{Preliminaries}

Given is a graph $G(V,E)$ with $|V|=n$ unless otherwise noted, used by travelers (or {\em people}) of two types: evaders (or {\em bads}) and innocents  (or {\em goods}). (These terms are used interchangeably.)
Each person $p$ can travel within some subgraph $G_p \subseteq G$. Depending on the setting, sensors can be placed on nodes or edges to capture the {\em flow} passing through.
A user $p$'s Markov chain determines the probability weight $f_{p,v}$ of $p$'s traffic through each node $v$. If oblivious, $p$ is unable to shift her flow $f_{p,v}$ from the path going through $v$ to some other path, so placing a sensor at $v$ captures all of $f_{p,v}$ (or at least whatever portion of it was not captured upstream). In some settings we assume all innocents, all evaders, or both are {\em oblivious}, as discussed below. 

We emphasize that {\em reactive} indicates a two-stage setting in which all the sensors are placed and then $p$ can choose an unblocked path in $G_p$ if one exists. Sensors are {\em not} deployed in sequence over time.
We also emphasize that person $p$ is restricted to subgraph $G_p$ regardless of whether $p$ is  oblivious or reactive, an evader or an innocent.

\vskip .1cm
\noindent \textbf{Edge and Node Interdiction.}
In {\em edge interdiction}, sensors are installed on edges and are represented by a matrix of decision variables $\mathbf{r}$: $r_{uv}=1$
if $(u,v)$ has a sensor placed at it (with cost $c_{uv}$) and $0$ otherwise. 
If an evader crosses an edge with a sensor she is detected with probability 1.
In {\em node interdiction}, placing a sensor on node $u$ (with cost $c_u$) means setting
$r_{uv}=1$ for every edge $(u,v)$, that is, interdicting all evaders {\em leaving} $u$). A sensor on a target node does {\em not} protect that node itself but will stop evaders as they pass through it.

The node and edge settings are equivalent in {\em general, directed graphs with location-varying costs}, in the sense that a problem in one setting can be transformed into the other \cite{Gutfraind09unreactive}.
Although the UME model is defined for convenience in terms of edge interdiction, unless otherwise stated we assume node interdiction.
\vskip .1cm
\noindent \textbf{Oblivious Evaders.}
An evader is specified in terms of the probabilities of her taking various routes,
where a route is a walk (possibly containing cycles) ending at a target node. 
A Markovian evader is represented by a Markov chain given by an initial
{\em source} distribution $\mathbf{a}$ over nodes and a transition
probability matrix $\mathbf{M}$. 
The matrix $\mathbf{M}$ has the property that a specified {target} node $t$ is an absorbing state: upon reaching $t$ the evader is removed from the network. 
Under mild restrictions on the Markov chain such as this, the probability of capturing the evader
can be expressed in closed form \cite{Gutfraind09unreactive}:
\begin{equation}
J({\bf \mathbf{a}},\mathbf{M},\mathbf{r})=1-\left({\bf \mathbf{a}}\left[\mathbf{I}-\left(\mathbf{M}-\mathbf{M}\odot{\bf \mathbf{r}}\right)\right]^{-1}\right)_{t}\label{eq:markov-evader-cost}
\end{equation}
where the symbol $\odot$ indicates element-wise (Hadamard) multiplication. 
This formulation generalizes to a setting of multiple simultaneous evaders, each realized with probability $w_e$,
or equivalently having weight $w_i$ representing the importance of capturing her. The probability of capturing  $e_i$ is denoted by $J_i(\mathbf{r})$.
\begin{definition}
An evader $e_i$ is specified by a $(\mathbf{M_i}, \mathbf{a_i})$ pair. Evader $e_i$ is {\em deterministic} if from each of her possible starting nodes, $\mathbf{M_i}$ specifies a single next node with probability $1$, and is {\em nondeterministic} otherwise. In both cases, $\mathbf{a_i}$ may specify multiple starting points with positive probability.
\end{definition}
\noindent \textbf{Budgeted Interdiction (BI).}
The BI objective is to capture as many evaders as possible, 
given a budget on sensors. 
More precisely, suppose we have a bound on the number of nodes we can monitor (or on their total cost, with costs always scaled to be integral). 
Any choice of some subset of nodes to observe determines a probability that a given evader will be captured 
(i.e., that she will pass through at least one observed node) prior to reaching her target $t$. 
The task in Budgeted Interdiction is to maximize the expected (weighted) number of evaders interdicted, subject to a budget $B$ on sensor costs:
$$\text{maximize}~ \sum_i w_i J_i(\mathbf{r}) ~\text{such that} ~\sum_{u} r_{u} c_{u} \le B$$

\noindent \textbf{Full Interdiction (FI).}
This problem seeks 
a minimum-cost set of nodes to observe in order to capture all evaders (denoted by $D$) with probability one. 
$$\text{minimize}~ \sum_{u} r_{u} c_{u}  ~\text{such that} ~ \sum_i J_i(\mathbf{r})= |D| $$
\vskip .1cm
\noindent \textbf{Reactive Evaders and Innocents.}
In this setting both types of travelers are reactive, which means a traveler is captured only if {\em all} her paths {\em within $G_p$} from source nodes to target node have received a sensor placed on some node prior to the target. Let $\hat w_p$ indicate the cost of having person $p$ succeed, which is negative if $p$ is good, and let binary variables $x_p=1$ indicate $p$'s success and $y_s=1$ indicate that bridge $s$ is open. Let $N$ indicate the goods and $D$ the bads. 
Then in the following formulation of the Bridges Problem the constraints code for the requirement that $x_p = \max_{s \in \sigma(p)} y_s$, i.e., traveler $p$ can cross iff at least one of her bridges is open, and the objective is to minimize the sum of errors (failed goods + successful bads).
$$\text{minimize}~ \sum_{p\in D} \hat w_p x_p+\sum_{p\in N} \hat w_p (1-x_p)  ~\text{such that} ~ x_p \le \sum_{s \in \sigma(p)} {y_s} ~\forall p ~\text{and}~ x_p \ge y_s ~\forall p, \forall s\in\sigma(p)$$

Let $TP,FP,TN,FN$ indicate the weighted numbers of true positives (bads failing to cross), false positives (goods failing), true negatives (goods succeeding), and false negatives (bads succeeding), respectively.
We focus on the objective of minimizing $FP+FN$,
in the problem setting in which each traveler's path from source to target passes through exactly one other node, a {\em bridge}. Since the bridges are the decision points, with {\em closed} corresponding to receiving a sensor and {\em open} corresponding to not, in this problem we use $n$ to denote the number of bridges.
The problem setting as defined generalizes to graphs in the roadblock model by replacing each bridge set $\sigma(p)$ with a set of paths to some target node. A traveler $p$ (either bad or good) is then captured if her target is disconnected from all her starting nodes in $G_p$ as for evaders before (where placing a sensor on a node means deleting it), and otherwise inflicts cost $\hat w_p$. 
Since each path is essentially a single edge, sensors can (as mentioned above) be thought of as either checkpoints or roadblocks.

\section{Oblivious innocents and evaders oblivious or not}

\subsection{Paths, trees, and cycles}
In this section we consider the Budgeted (BI) and Full Interdiction (FI) problems where the graph $G$ (on $n$ nodes) is one of several special topologies.

\begin{definition}
In a path graph $P$ with nodes numbered $1$ through $n$, an {\em interval} $[x,y]$ indicates the sequence of nodes numbered $x$ through $y$ (with $x \le y$)  the interval's {\em startpoint} and {\em endpoint}, respectively. Half-open intervals $[x,y) = [x,y-1]$ and $(x,y] = [x+1,y]$ are defined similarly. For nodes $x,y$ we write $x<y$ to indicate that $x$ precedes $y$ in $P$. Similarly, in a tree $T$, an {\em interval} $[x,y]$ is the sequence of nodes lying on the path in $T$ from $x$ to $y$. A node $v$ {\em pierces} interval $[x,y]$ if $v \in [x,y]$.
An {\em interval sequence} is a set of intervals that can be ordered so that each interval is strictly contained by the previous one. 
All the intervals in a {\em suffix sequence} share the same endpoint; all the intervals in a {\em prefix sequence} share the same start point.
\end{definition}

The budgeted problem is solvable in polynomial time as we show below. It will follow that FI is also solvable 
by searching for the smallest budget $B$ that permits full interdiction. The problem can be solved more efficiently, however.

\begin{theorem}\label{thm:optonpaths}
When $c_v=1 ~\forall v$, Full Interdiction is optimally solvable in $O(n \log n)$ time on paths.
\end{theorem}
\begin{proof}
Consider an evader $e_i$ with start nodes $S_i$ and a target node $t_i$. 
We must capture evader $e_i$ in the case of each starting point $s \in S_i$ before she reaches node $t_i$. 
Node $s$ lies either to the left or right of $t_i$, assume to the left, i.e., $s < t_i$ (e.g., node 3 in Fig. \ref{fig:stabbing}). 
Evader $e_i$ may (probabilistically) move to the left before returning right, 
and so a sensor placed to the left of $s$ may capture the evader with positive probability. 
For capturing $e_i$ with probability 1, however, it is necessary and sufficient to place a sensor somewhere in the interval $[s,t_i)$ ($[3,6)$ for the first evader $e_1$ in Fig. \ref{fig:stabbing}).

Each starting point $s$ of evader $e_i$ will correspond to an interval $[s,t_i)$ or interval $(t_i,s]$, 
depending on the relative location of $s$ to $t_i$. 
Each such interval must be pierced.
Intervals of the former kind (with the evader approaching the target from the left) will form a sequences of suffix intervals;
intervals of the latter kind (with the evader is approaching from the right) will form a sequence of prefix intervals. 
In the worst case, it could be necessary to consider $m=O(n^2)$ intervals
because each interval may be traversed with positive probability by some evader.
It suffices to consider each evader's smallest left interval and smallest right interval ($[3,6)$ and $(6,8]$ for $e_1$ in Fig. \ref{fig:stabbing})), since each such interval is contained within all others in the sequence.
We build an {\em interval graph} $H$ by associating a node with each smallest interval (each of which can be found in time $O(\log n)$ by binary search)
and placing an undirected edge for any two smallest intervals that intersect.  
Because the cost of piercing any interval is $c_v=1$, and because each intersection of intervals corresponds to a clique of $H$,
Full Interdiction is equivalent to Minimum Clique Cover on $H$.
The latter is solvable in linear time on the interval graph (plus time for sorting) \cite{EvenLRSSS08}. 
\qed
\end{proof}

\begin{figure}[t!]
\centering \includegraphics[width=.55\textwidth]{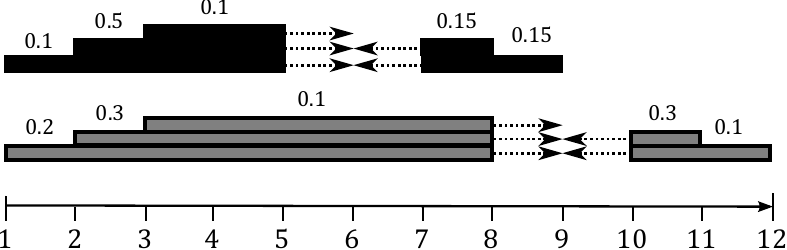} 
\caption{An instance of Network Interdiction with two stochastic evaders on the graph $P_{12}$. One evader is traveling from nodes 3 and 8 to 6, and one is traveling from nodes 3 and 11 to 9.
Because of their stochastic motion the evaders could be partially interdicted at nodes such as $1$ or $12$ that do not lie on the shortest paths to their targets.}\label{fig:stabbing}
\end{figure}

\noindent A generalization is also possible, as follows.

\begin{theorem}
When $c_v=1 ~\forall v$, Full Interdiction is optimally solvable in $O(n^3)$ time on trees.
\end{theorem}

\begin{proof}
The $O(n^2)$ intervals are now paths in the tree, whose intersection graph (constructable in $O(n^3)$) is a {\em chordal graph}, on which Minimum Clique Cover can also be solved in linear time \cite{Gavril72}.  \qed
\end{proof}

For the setting of edge sensors, Full Interdiction is closely related to the Minimum Directed Multicut (MDM) problem,
in which the task is to find a minimum cut that separates each of $k$ source-sink pairs $(s_i,t_i)$.
\begin{proposition}
In the edge interdiction setting, Full Interdiction is 2-approximable on trees, which is the optimal factor (assuming the Unique Games Conjecture \cite{KhotR08}).
\end{proposition}
\begin{proof}
When restricted to an underlying tree graph, the Full Interdiction problem is identical to Directed Multicut. \qed
\end{proof}

\begin{definition}
For a possible route $r$ traveled by some evader, 
let $V_r$ indicate the nodes visited along route $r$ {\em before} reaching its target, or the {\em route set} of $r$. Let $m$ be the number of distinct route sets among all evaders.
\end{definition}

Note that multiple distinct routes can give rise to the same route set, and that a route set in a path graph is always an interval with an end point at the target node. We now turn to Budgeted Interdiction. 
\begin{theorem}\label{thm:detdp}
Let $m$ be the total number of different evader route sets.
Budgeted Interdiction with deterministic evaders and (integer) budget $B$ is optimally solvable on the path graph in time $O(B n m) = O(B n^3)$.
\end{theorem}
\begin{proof}
We give a dynamic programming solution in Algorithm \ref{alg:dp}. 
We compute an optimal solution using a table $opt[\ell,\hat{v},b]$
that stores the optimal solution restricted to the $\ell$ left-most intervals, nodes $1...\hat{v}$ and budget $b$.
We first compute the value of node $v$ restricted to the first $\ell$ intervals,
i.e., $val[\ell,v]$ is the sum of the weights of those intervals when the only sensor is node $v$.
Each subproblem solution is computed in constant time: given inputs $\ell,v,b$, 
if $v$ is not chosen, then the optimal solution value is the same as inputs $\ell,v-1,b$; 
if $v$ is chosen, then the optimal solution value is the value of choosing $v$ in this situation, 
plus optimal solution on the intervals lying to the left of $v$, 
using the first $v-1$ nodes and a budget of $b-c_v$ (or 0 if $b-c_v<0$).

Proof of correctness is by induction: if node $v$ is chosen, then due to the linear ordering, nodes prior to $v$ only contribute to piercing intervals $1$ through $pr(v)$.
Note that correctness holds also when interval weights may be negative.
\qed
\end{proof}

\begin{algorithm}[t!]
\caption{Budgeted Interdiction DP for Evaders on the Path Graph} \label{alg:dp}
\begin{algorithmic} [1]
\STATE sort the $O(n^2)$ intervals by right endpoint
\STATE $pr[v]$ = index of the last interval lying before node $v$, or 0 if none \textbf{for} every $v$
\STATE $val[\ell,v]$ = value of node $v$, restricted to intervals 1 to $\ell$, \textbf{for} every $v,\ell$
\STATE $opt[0,v,b]$ = 0 \textbf{for} every $v,b$
\STATE $opt[\ell,0,b]$ = 0 \textbf{for} every $\ell,b$
\STATE $opt[\ell,v,0]$ = 0 \textbf{for} every $\ell,v$
\FOR {$b$ = 1 to $B$}
	\FOR {$\ell$ = 1 to $m$}
		\FOR {$v$ = 1 to $n$}
      			\STATE $opt[\ell,v,b] = \max\{opt[\ell,v-1,b],~ val[\ell,v] + opt[pr[v], v-1, \max(0,b-c_v)]\}$
		\ENDFOR
	\ENDFOR
\ENDFOR
\STATE \textbf{return} $opt[m,n,B]$
\end{algorithmic}
\end{algorithm}

The case of nondeterministic evaders is more complicated since, as noted above, it gives rise to sequences of suffix intervals and sequences of prefix intervals. For each such sequence corresponding to a single nondeterministic evader, the computation of $val[\ell,v]$ will be based on all the intervals in the sequence that $v$ pierces. More precisely, let $\{[1,t),[2,t),...,[s,t)\}$ be a suffix sequence for some nondeterministic evader $e_i$ with source $s$ and target $t$. For each node $v<t$ there is some probability $p_v$ that placing a sensor at node $v$ {\em suffices} for capturing $e_i$. Namely, $p_v$ is 1 for any $v \in [s,t)$, while for each node $v<s$ the probability $p_v$ can be computed based on the Markov chain of $e_i$, that is,
just the probability that her Markov chain visits $v$ and is computed as follows.  
For $e_i$'s Markov chain $(\mathbf{a}, \mathbf{M})$, let $\mathbf{M_{-v}}$ denote a transition matrix where row $v$ has been replaced by zeros,
i.e. the chain with $v$ as a killing state.
Then $p_v = \left({\bf \mathbf{a}}\left[\mathbf{I}-\mathbf{M_{-v}}\right]^{-1}\right)_{v}$.

For each interval in the sequence, we now define a {\em marginal} probability $\hat p_v$ as follows: $\hat p_1 = p_1$; $\hat p_v =  p_v - p_{v-1}$ for $1 < v \le s$, and $\hat p_v = 1 - p_{s-1}$ for $s \le v < t$. By construction, the $\hat p_v$ values for all intervals containing a given node $u$ will sum to exactly the probability of evader $e_i$ reaching node $u$, and hence of such a sensor placement sufficing to capture evader $e_i$. (The values labeling the intervals in Fig. \ref{fig:stabbing} are the marginal probabilities, weighted by the probability of choosing their starting points.) Marginal probabilities are assigned to prefix intervals similarly. Therefore the value of a set of sensor locations for a given instance of the problem with nondeterministic evaders is exactly the value of those locations for the resulting problem instance with interval sequences of deterministic evaders; that is, the nondeterministic problem reduces to the deterministic problem (albeit with up to a factor $n$ more intervals). Thus we have the following.

\begin{theorem}\label{thm:nondeterm}
Budgeted Interdiction with nondeterministic evaders is optimally solvable on the path graph, in time $O(B n^2 m) = O(B n^4)$.
\end{theorem}

These problems can also be solved on the cycle by reduction to path graphs.
\begin{theorem}
Full Interdiction is optimally solvable in $O(n^2)$ time on the cycle graph. Budgeted Interdiction with deterministic or nondeterministic evaders and budget $B$ is optimally solvable on the cycle graph in time $O(Bn^4)$ or $O(Bn^5)$, respectively.
\end{theorem}
\begin{proof}
For the minimization problem, we reduce to a collection of $n$ path graph instances, corresponding to $n$ ways to ``cut'' the cycle graph, as follows. 
For each node $v \in V$, consider placing a sensor at node $v$. 
It will pierce some set of intervals, 
with the effect that none of the remaining intervals to pierce include node $v$, 
yielding a path graph instance with nodes $v+1,...,n,1,...,v-1$. 
Solve each resulting path graph instance in linear time, and return the cheapest solution (combined with $v$). 
The budgeted problems are solved by a similar reduction. \qed
\end{proof}
The process can be generalized to Full Interdiction on arbitrary graphs containing $c$ cycles, though at a cost of $O(n^c)$: find all the cycles \cite{Johnson75} and then explore all possible cuts.

\subsection{General graphs}
We show that Budgeted Interdiction (BI) is hard already with one Markovian evader.
This improves on the result in \cite{Gutfraind09umek} which held for two or more evaders.
\begin{theorem}\label{thm:singleev}
Budgeted Interdiction is NP-hard even with a single Markovian evader.
\end{theorem}
\begin{proof}
We reduce from Vertex Cover (VC) to the decision problem of determining whether 
the interdiction probability $J$ can be raised to a certain threshold using at most $B$ sensors. 
Given a VC problem instance, i.e., a graph $G$ on $n$ nodes and an integer $B$,
we construct a network interdiction (NI) instance with a Markovian evader on a graph $G'$. 
The graph $G'$ extends graph $G$ by adding a target node $t$, 
which is made adjacent to all other nodes. 
We define the evader $e$ thus. 
Each node corresponds to a state of its Markov chain. 
All non-target nodes are equally likely to be chosen as $e$'s start node. 
When at a given node $v$, $e$ moves to the target $t$ with probability $50\%$; 
otherwise, $e$ moves to one of $v$'s other neighbors, chosen uniformly at random. 

For a particular solution, let the {\em profit} for a node be the probability of interdiction if the evader starts at that node.
We will now show that the VC instance admits a vertex cover of size $B$ iff the NI instance admits a size-$B$ solution of profit at least $B+(n-B)/2$.
Note that an overall interdiction probability of $\frac{(n+B)}{2n}$ is the same as a total profit of $(n+B)/2 = B+(n-B)/2$ over all nodes.

First assume there is a size-$B$ vertex cover $C$ of $G$. 
Then an NI solution with sensors placed at all the nodes in $C$ will have profit $B+\frac{(n-B)}{2}$:
1 for each of the $B$ nodes in $C$ plus 1/2 for each of the remaining $n-B$ nodes, 
since for any node $v$ not in $C$, all $v$'s neighbors in $G$ must be in $C$.

Now assume there is no size-$B$ vertex cover, and consider a set $S$ of $B$ nodes, a set which must fail to cover some edge. 
Again for each of the $B$ nodes in $S$ we have profit 1. 
{Every} other node $v$ will have profit at most 1/2, since without its own sensor, 
an evader starting at $v$ goes directly to $t$ with probability 1/2. 
But now consider an edge $(u,v)$ that is left uncovered by $S$. 
The evasion probability when starting at $u$ is greater than 1/2---at least $1/2 + 1/(4\deg(G))$---since if $e$ reaches node $v$, 
it now has a second chance to move to $t$, and so the profit of $u$ is less than 1/2. 
Therefore the total profit is strictly less than $B+(n-B)/2$. 
\qed
\end{proof}

It follows from the hardness proof of \cite{Gutfraind09umek} that Full Interdiction is NP-hard with 2 evaders. 
It does not remain hard when limited to a single evader, however.

\begin{theorem}\label{thm:mincutred}
Full Interdiction with one evader is solvable in polynomial time.
\end{theorem}
\begin{proof}
We solve the problem by reducing to a Min Cut problem. 
Given a set of routes specifying the evader's behavior, we introduce a source node $s$ pointing to all start nodes of its Markov chain. 
All edges that the evader has zero probability of reaching and crossing are removed from the graph $G$. 
Any unreachable nodes are also removed. 
Now, in order to interdict the evader before they reach $t$, we must delete vertices in order to separate $s$ from $t$ in $G$. 
It is well known that this Min Vertex Cut problem can be solved in polynomial time, 
by reduction to Directed Min Cut, as follows \cite{Even79}. 
First replace any undirected edge with a pair of directed edges. 
Then replace each node $v$ (other than $s$ or $t$) with a pair of nodes and directed edge $(v_a,v_b)$, 
where each edge directed to $v$ is now directed to $v_a$ and each edge directed from $v$ is now directed from $v_b$. We compute a Min Cut on the resulting graph $G'$. 
If any edge is chosen that does not correspond to a node in $G$, 
we can substitute one of the edges corresponding to its two vertices (if one of these is the target, then the non-target node is chosen). 
The resulting modified Min Cut solution to $G'$  will correspond to a Min Vertex Cut solution to $G$, 
and moreover to a Full Interdiction solution.
\qed
\end{proof}

We now turn to approximation algorithms for the general setting, by relating interdiction to the Set Cover and Maximum Coverage problems. It was shown in \cite{Gutfraind09unreactive} that weighted Budgeted Interdiction with any number of Markovian evaders is 1-1/e-approximable, which is the optimal factor (see Appendix \ref{app:hardness}).

Identifying nodes and route sets with elements and sets in the Hitting Set problem yields a reversible reduction, 
and hence the following immediately results:

\begin{corollary}
Full Interdiction is hard to approximate
with factor $(1-\epsilon)\ln m$ for any $\epsilon>0$, assuming $NP \subseteq DTIME(m^{O(\log \log m)})$ but can be approximated with factor $H_m$ in time polynomial in $n+m$.
\end{corollary}
\begin{proof}
We reduce from Set Cover, as in \cite{Gutfraind09unreactive}, creating a node for each set and a route set (with a corresponding deterministic evader) for each element. 
\qed
\end{proof}

\section{Reactive innocents and the Bridges problem}

Since maximizing the ``net flow'' $TN-FN$ \cite{Rubinstein} turns out to be as hard to approximate as Maximum Independent Set (see Appendix \ref{app:hardness}), we focus primarily on the the min-error $FP+FN$ setting. A geometric or ``convex'' version of the min-error problem is optimally solvable, however. Since the two objective functions differ only by a constant and a negation ($TN-FN = (w_N-FP)-FN$, where $w_N$ indicates the total value of all goods), the same holds for the net flow problem.

\subsection{Convex bridge sets\label{subsec:convex}}

\begin{definition}
An instance of the Bridges Problem is {\em convex} if the bridges can be ordered so that if two bridges $x$ and $y$ are accessible to a person $p$ then any bridge $z$ with $x < z < y$ is accessible to $p$ as well.
\end{definition}
The problem example shown in Figure~\ref{fig:bridges} in the introduction is convex. We assume that the indices of people are sorted in order of their positions from left to right and the bridge indices are sorted in order of their rightmost accessing person.
This setting can be solved by mapping it to Budgeted Interdiction on the path graph and adapting Algorithm \ref{alg:dp}.
\begin{corollary}
The convex Bridges problem is solvable in time $O(n|N|+n|D|) = O(n^3)$.
\end{corollary}
\begin{proof}
Given a Bridges problem instance (say in the min-error formulation), we introduce a Budgeted Interdiction instance (with budget arbitrarily large) as follows.
Each of the bridges is identified with a node on the path graph.
For each traveler $p$ we define an evader $p$ on the interval $I_p$, where $I_p $ are all bridges available to $p$.
This produces $|N|+|D| \le n(n-1)$ distinct intervals.
The weight of evader $p$ is set to negation of the traveler's cost: $w_p = -\hat w_p$.

The formulations are now equivalent: in the Bridges Problem, a traveler succeeds iff one or more of her bridges is open; 
in BI, an evader is interdicted iff one or more of the nodes in her interval is interdicted by a sensor.

Then we pass the instance to an adaptation of Algorithm \ref{alg:dp}:
we remove the budget dimension from the dynamic programming table
and also remove the outer loop iterating over budget values, saving a factor of $O(b)$ in running time. 
The resulting algorithm computes an optimal interdiction solution.
(Recall that Algorithm \ref{alg:dp} supports intervals with weights both positive and negative.)
Given this solution, we then solve the Bridges problem by opening a bridge iff the corresponding node has a sensor placed at it.
\qed
\end{proof}

\subsection{The Min-error $FP+FN$ setting}

NP-hardness of optimally solving the min-error setting follows from the hardness of the net-flow setting: maximizing $TN-FN$ is the same as minimizing $FN-TN = FN - (|N|-FP)=FP+FN-|N|$. The hardness of approximation properties, however, are not the same. In fact, the min-error problem is precisely the Positive-Negative Partial Set Cover
Problem \cite{Miettinen08}, which, as a generalization of Red-Blue Set Cover, is strongly inapproximable (hard to approximate with factor $\Omega(2^{\log^{1-\epsilon}}m))$ (where $m$ is the number of sets) unless $NP \subseteq DTIME(m^{\text{polylog}(m)})$ though approximable with factor $2\sqrt{(m+\pi) \log \pi}$, where $\pi$ is the number of goods.

Glazer \& Rubinstein define what we will call a {\em claw} as an object $c$ consisting of a good $g_c$ and minimal set of bads $B_c$ such that for each bridge $s \in \sigma(g_c)$, $s$ is also in $\sigma(b_i)$ for some bad $b_i \in B_c$, which means that in any consistent solution, either $g_c$ must fail or at least one $b_i$ must succeed. They show that this is also a sufficient condition for being a valid solution, and hence obtain a
Set Cover problem: for each claw, choose a person to err on, with minimum total error cost over all claws. Unfortunately, this instance in general has exponentially many constraints (since for each good $g$ with bridge set $\sigma(g)$, each of whose bridges admit some number $bads(s)$ of bads, there will be $|C| = \Pi_{s \in \sigma{g}} bads(s)$ many claws), and so the $O(\log |C|)$ approximability of set cover becomes trivially weak. We therefore modify the definition of claw slightly as follows.

\begin{definition}
A {\em claw} is an object $c$ consisting of a good $g_c$ and, for each bridge $s_i \in \sigma(g_c)$ the set $b \in \sigma^{-1}(s_i)$ of all the bads who can use bridge $s_i$.
\end{definition}

Each claw $c$ therefore imposes the following constraint: in any valid solution, either $g_c$ must fail or all the bads in $\sigma^{-1}(s_{i})$ for some $s_i \in \sigma(g_c)$ must succeed. Given $c$, let a {\em kill move} be the action of killing $g_c$; let an {\em open bridge move} be the action of opening some bridge $s_{i}$. Now we can interpret this problem as an instance of Submodular Cost Set Cover \cite{KoufogiannakisY09,IwataN09} in which the elements are claws and there are two kinds of sets. For each possible kill move $m_g$, introduce a set $M_g = \{g\}$;
for each possible open bridge move $m_{{i}}$, introduce a set $M_{i}$ 
consisting of all the claws that opening bridge $i$ would satisfy. 
There are $N$ elements (claws) and $N+m$ sets (moves).

\begin{theorem}\label{thm:maxapprox}
The general $FP+FN$ Bridges problem is $(1+\max_{g \in N} |\sigma(g)|)$-approximable.
\end{theorem}
\begin{proof}
First we claim that the cost of a set of moves is submodular. Indeed, the cost of each kill move is simply the additive cost of the specified good failing; the marginal cost of an open bridge move is monotonically decreasing since it is based on the number of {\em additional} bads that opening the bridge then allows to succeed.
Second we claim that the value of the total error of the Bridge solution returned is at most the cost of the moves chosen. Indeed, first, the only time bridges are opened is during bridge moves, and so the total cost of bads succeeding is at most the cost of the open bridge moves; second, when bridges are closed at the end, all constraints have been handled, and so the failures of all goods have already been ``paid for'', in the cost of the kill moves.
Therefore the algorithms of \cite{KoufogiannakisY09,IwataN09} apply, which provide a solution with approximation factor $f$, which is the maximum number of sets that any element appears in. In the constructed set cover instance, $f$ translates into 1 plus $\max_{g \in N} |\sigma(g)|$.
\qed
\end{proof}

\vskip .25cm
\small
\noindent \textbf{Acknowledgments.}
We thank Amotz Bar-Noy and Rohit Parikh for useful discussions. This work was funded by the Department of Energy at 
the Los Alamos National Laboratory through the LDRD program, 
and by the Defense Threat Reduction Agency. Indexed as Los Alamos Unclassified Report LA-UR-11-10123.

\normalsize

\appendix

\section{Other hardness results}\label{app:hardness}

The following two results are approximation-preserving reductions from the Maximum Independent Set (MIS) problem, which is hard to approximate with factor  $n^{1-\epsilon}$  (where $|V|=n$) for any $\epsilon>0$ \cite{Zuckerman07}. A MIS instance consists of a graph $G = (V,E)$ and a positive integer $k$.

\begin{proposition}
The Bridges problem variant in which the goal is to maximize $TN$ subject to a bound on $FN$ is NP-hard to approximate with factor $n^{1-\epsilon}$.
\end{proposition}
\begin{proof}
In our reduction, each vertex $v$ becomes a bridge $s_v$ and a bad $b_v$ who can cross only $s_v$. Each edge $(u,v)$ becomes $k+1$ goods who can cross bridges $s_u$ and $s_v$. The bound on $FN$ is set to $k$, which prevents any two goods connected by an edge from both failing. \qed
\end{proof}

\begin{proposition}\label{inthardness}
The net-flow $TN-FN$ setting of the Bridges problem is NP-hard to approximate with factor $n^{1-\epsilon}$.
\end{proposition}
\begin{proof}
In our reduction, each vertex becomes a bridge and (usually) some bad people, and each edge becomes a good person. All the people introduced have value 1 or -1. Specifically, for each vertex $v \in V$, we introduce a bridge $s_v$ and $deg(v)-1$ bads, whose only accessible bridge is $s_v$ itself. (If $deg(v)=0$, we similarly introduce one good.) For each edge $e=(u,v) \in E$, we introduce a good $p_e$, whose accessible bridges are $u$ and $v$.

We now claim that the MIS instance has a solution of value at least $k$ iff the Bridges problem instance does. First, assume there is an independent set $S$ of size $k$.  For each vertex $v \in S$, we open the corresponding bridge. Each bridge $s_v$ with $deg(v)=0$ has one good and no bads, for a net value of 1. For each bridge $s_v$ with $deg(v)>0$, there are $deg(v)$ goods who can cross it (and possibly others) and $deg(v)-1$ bads who can cross only it. Since no two vertices in $S$ are adjacent, though, for each open bridge the goods who can cross it can cross no other open bridges. Therefore for each open bridge $s_v$, all its $deg(v)$ goods will use it, which means that bridge contributes exactly $deg(v)-(deg(v)-1)=1$ to the solution value, for a total of $k$.

Conversely, assume there is a bridges solution of value at least $k$. Observe that no open bridge can contribute value greater than 1, since at most $deg(v)$ goods use it but necessarily all its $deg(v)-1$ bads will do so. Therefore a solution of value $k$ will involve opening at least $k$ bridges. If any bridge can be closed without decreasing the solution quality, do so, repeatedly, until there is no longer any such bridge. At that point, the solution will consist of $k$ open bridges, each of value $k$. But again by the previous argument, in order for two bridges each to contribute value 1, the corresponding vertices must be independent. Thus the $k$ vertices corresponding to the open bridges form an independent set.
\qed
\end{proof}

\begin{corollary}\label{mipprob}
The {maximum-probability} net-flow $TN-FN$ setting of the Bridges problem of \cite{Rubinstein} is no easier than the ({integral}) net-flow $TN-FN$ setting.
\end{corollary}
\begin{proof}
Consider the Bridges problem instance produced in Proposition \ref{inthardness}, but now allow fractional bridge openings and take the max-probability objective. An integral solution is in particular a valid fractional solution, and so the forward direction of the iff goes through unchanged. Now assume there is a max-prob fractional solution of value at least $k$. Suppose some bridge $s$ is open with probability $p$, $0<p<1$. If more bads are using $s$ than goods, then closing $s$ will only improve the solution, so assume otherwise. In this case, assume that $\gamma$ goods are using $s$ and $\beta$ bads are, with $\gamma \ge \beta$. Then fully opening the bridge will increase the value of {\em at least} $\gamma$ goods by amount $(1-p)$---any other goods that had chosen other bridges that were also open with probability $p$ will now shift to this bridge---and will increase the value of $\beta$ bads by the same amount $(1-p)$, for a total change to the bridge's net flow of at least $(1-p)\gamma-(1-p)\beta$, which is non-negative. Therefore we can convert the fractional solution into an integral solution of value still at least $k$. But then by the previous argument we can use the solution to obtain an independent set of size $k$. \qed
\end{proof}

\begin{corollary}
The {maximum-probability} min-error $FP+FN$ setting of the Bridges problem of \cite{Rubinstein} is no easier than the ({integral}) min-error $FP+FN$ setting.\end{corollary}
\begin{proof}
The proof is similar to that of Corollary \ref{mipprob}. \qed
\end{proof}

\begin{proposition}\label{prop:budgetedhardness}
The Budgeted Interdiction problem in NP-hard to approximate within factor $1-1/e-\epsilon$ for any $\epsilon>0$.
\end{proposition}
\begin{proof}
We reduce from Maximum Coverage, which has the stated hardness property \cite{Feige98}.

Given is a family of subsets $S_i$ of a ground set $U=\{e_1,...,e_n\}$. 
The task is to choose $k$ subsets whose union is of maximum cardinality. 
For each set $S_i$ we introduce a corresponding node $v_i$. 
For each element $e_j$ we introduce a corresponding evader whose Markov chain takes it deterministically 
(in some arbitrary order) through all the nodes corresponding to sets containing $e_j$ and thence to a special target node. 
Then a selection of sets covering evader paths is equivalent to a selection of sets covering elements, with exactly the same solution value. 
\qed
\end{proof}

\bibliographystyle{abbrv}
\bibliography{bib,interdict}

\begin{thebibliography}{10}

\bibitem{Bar-noy-1995-complexity}
A.~Bar-Noy, S.~Khuller, and B.~Schieber.
\newblock The complexity of finding most vital arcs and nodes.
\newblock Technical report, University of Maryland, College Park, MD, USA,
  1995.

\bibitem{Corley-1982-most}
H.~W. Corley and D.~Y. Sha.
\newblock Most vital links and nodes in weighted networks.
\newblock {\em Operations Research Letters}, 1(4):157 -- 160, Sep 1982.

\bibitem{EvenLRSSS08}
G.~Even, R.~Levi, D.~Rawitz, B.~Schieber, S.~Shahar, and M.~Sviridenko.
\newblock Algorithms for capacitated rectangle stabbing and lot sizing with
  joint set-up costs.
\newblock {\em ACM Transactions on Algorithms}, 4(3), 2008.

\bibitem{Even79}
S.~Even.
\newblock {\em Graph Algorithms}.
\newblock Computer Science Press, 1979.

\bibitem{Feige98}
U.~Feige.
\newblock A threshold of ln for approximating set cover.
\newblock {\em J. ACM}, 45(4):634--652, 1998.

\bibitem{Gavril72}
F.~Gavril.
\newblock Algorithms for minimum coloring, maximum clique, minimum covering by
  cliques, and maximum independent set of a chordal graph.
\newblock {\em SIAM J. Computing}, 1(2):180--187, 1972.

\bibitem{Rubinstein}
K.~Glazer and A.~Rubinstein.
\newblock A study in the pragmatics of persuasion: A game theoretical approach.
\newblock {\em Theoretical Economics}, 1:395--410, 2006.

\bibitem{Gutfraind09umek}
A.~Gutfraind and K.~Ahmadizadeh.
\newblock {Markovian Network Interdiction and the Four Color Theorem}, 2009.
\newblock In review with {SIAM J. Discrete Math.}
  http://arxiv.org/abs/0911.4322.

\bibitem{Gutfraind11markovian}
A.~Gutfraind, A.~Hagberg, D.~Izraelevitz, and F.~Pan.
\newblock {Interdiction of a Markovian Evader}.
\newblock In R.~Dell and K.~Wood, editors, {\em {Proc. INFORMS Computing
  Society Conference}}, Jan 2011.

\bibitem{Gutfraind09unreactive}
A.~Gutfraind, A.~Hagberg, and F.~Pan.
\newblock Optimal interdiction of unreactive {M}arkovian evaders.
\newblock In J.~Hooker and W.-J. van Hoeve, editors, {\em {Proc. CPAIOR}},
  volume 5547 of {\em Lecture Notes in Computer Science}, pages 102--116.
  Springer, May 2009.

\bibitem{Corley1974}
J.~H.~W.~Corley and H.~Chang.
\newblock Finding the n most vital nodes in a flow network.
\newblock {\em Management Science}, 21(3):362--364, November 1974.

\bibitem{IwataN09}
S.~Iwata and K.~Nagano.
\newblock Submodular function minimization under covering constraints.
\newblock In {\em FOCS}, pages 671--680, 2009.

\bibitem{Johnson75}
D.~B. Johnson.
\newblock Finding all the elementary circuits of a directed graph.
\newblock {\em SIAM Journal on Computing}, 4(1):77--84, 1975.

\bibitem{KhotR08}
S.~Khot and O.~Regev.
\newblock Vertex cover might be hard to approximate to within $2-\epsilon$.
\newblock {\em J. Comput. Syst. Sci.}, 74(3):335--349, 2008.

\bibitem{KoufogiannakisY09}
C.~Koufogiannakis and N.~E. Young.
\newblock Greedy \ensuremath{\Delta}-approximation algorithm for covering with
  arbitrary constraints and submodular cost.
\newblock In {\em ICALP (1)}, pages 634--652, 2009.

\bibitem{Mcmasters-1970-optimal}
A.~W. McMasters and T.~M. Mustin.
\newblock Optimal interdiction of a supply network.
\newblock {\em Naval Research Logistics Quarterly}, 17(3):261--268, 1970.

\bibitem{Megiddo83}
N.~Megiddo, E.~Zemel, and S.~L. Hakimi.
\newblock The maximum coverage location problem.
\newblock {\em SIAM Journal on Algebraic and Discrete Methods}, 4(2):253--261,
  1983.

\bibitem{Miettinen08}
P.~Miettinen.
\newblock On the positive-negative partial set cover problem.
\newblock {\em Inf. Process. Lett.}, 108(4):219--221, 2008.

\bibitem{Pan-2003-models}
F.~Pan, W.~S. Charlton, and D.~P. Morton.
\newblock Interdicting smuggled nuclear material.
\newblock In D.~Woodruff, editor, {\em Network Interdiction and Stochastic
  Integer Programming}, pages 1--19. Kluwer Academic Publishers, Boston, 2003.

\bibitem{Peleg07}
D.~Peleg.
\newblock Approximation algorithms for the label-cover$_{\mbox{max}}$ and
  red-blue set cover problems.
\newblock {\em J. Discrete Algorithms}, 5(1):55--64, 2007.

\bibitem{Ratliff1975}
H.~D. Ratliff, G.~T. Sicilia, and S.~H. Lubore.
\newblock Finding the n most vital links in flow networks.
\newblock {\em Management Science}, 21(5):531--539, January 1975.

\bibitem{Zuckerman07}
D.~Zuckerman.
\newblock Linear degree extractors and the inapproximability of max clique and
  chromatic number.
\newblock {\em Theory of Computing}, 3(1):103--128, 2007.

\end{thebibliography}

\end{document}